\documentclass[runningheads]{llncs}
\newif\iflong
\newif\ifshort

\longtrue

\iflong
\else
\shorttrue
\fi
\usepackage[T1]{fontenc}
\usepackage{graphicx,thm-restate}
\usepackage{xspace}
\usepackage[hidelinks]{hyperref}
\newcommand{\fcc}{\textsc{\textup{Fair Correlation Clustering}}\xspace}
\newcommand{\fccshort}{\textsc{\textup{FCC}}\xspace}

\newcommand{\tc}{\tilde{c}}

\newcommand{\set}[1]{{\{#1\}}}

\usepackage{amsfonts,amssymb,mathtools}
\usepackage{complexity}
\usepackage{algorithm,environ,tabularx,mdframed,xpatch,xspace,todonotes}
\usepackage{algpseudocode}
\usepackage[dvipsnames]{xcolor}
\usepackage{tikz}
\usetikzlibrary{calc,arrows,shapes,shapes.geometric,matrix,decorations.pathreplacing,positioning,patterns,backgrounds,fit,topaths,mindmap,shapes.callouts,decorations.text,chains,intersections,hobby}
\usetikzlibrary{svg.path}

\makeatletter
\renewcommand{\@Opargbegintheorem}[4]{%
  #4\trivlist\item[\hskip\labelsep{#3#2\@thmcounterend}]}
\makeatother

\newcommand{\clos}{\ensuremath{\text{clos}}}
\newcommand{\bigoh}{\mathcal{O}}

  \spnewtheorem{observation}{Observation}{\bfseries}{}

\newcommand{\comp}{\mathsf{comp}}
\newcommand{\sol}{\mathsf{cluster}}
\newcommand{\Comps}{\mathsf{Comp}}
\newcommand{\Sols}{\mathsf{Clusters}}
\newcommand{\Cuts}{\mathsf{Cuts}}
\newcommand{\cut}{\mathsf{cut}}
\newcommand{\inbound}{\zeta}
\newcommand{\solbound}{\gamma}

\newcommand{\N}{\mathbb{N}}

\newcommand{\yesinstance}{\textup{YES}-instance\xspace}

\DeclareMathOperator{\vcn}{vcn}
\DeclareMathOperator{\td}{td}
\DeclareMathOperator{\tw}{tw}
\DeclareMathOperator{\counter}{count}
\DeclareMathOperator{\cost}{cost}
\DeclareMathOperator{\open}{open}
\DeclareMathOperator{\current}{current}
\DeclareMathOperator{\size}{size}

\makeatletter
\newcommand{\problemtitle}[1]{\gdef\@problemtitle{#1}}
\newcommand{\probleminput}[1]{\gdef\@probleminput{#1}}
\newcommand{\problemtask}[1]{\gdef\@problemtask{#1}}

\NewEnviron{myproblem}{
\problemtitle{}\probleminput{}\problemtask{}
\BODY
\noindent
\begin{tabularx}{\textwidth}{
	 						l X c}
	\multicolumn{2}{
						l}{\@problemtitle} \\
	\textbf{Input:} & \@probleminput \\
	\textbf{Question:} & \@problemtask
\end{tabularx}
}
\makeatother

\makeatletter
\xpatchcmd{\endmdframed}
  {\aftergroup\endmdf@trivlist\color@endgroup}
  {\endmdf@trivlist\color@endgroup\@doendpe}
  {}{}
\makeatother

\surroundwithmdframed[linecolor=black,linewidth=1pt,skipabove=1mm,skipbelow=1mm,leftmargin=0mm,rightmargin=0mm,innerleftmargin=1mm,innerrightmargin=1mm,innertopmargin=1mm,innerbottommargin=1mm]{myproblem}
\surroundwithmdframed[linecolor=black,linewidth=1pt,skipabove=1mm,skipbelow=1mm,leftmargin=0mm,rightmargin=0mm,innerleftmargin=1mm,innerrightmargin=1mm,innertopmargin=1mm,innerbottommargin=1mm]{problem_compact}

\tikzset{matrixstyle/.style={matrix of nodes, ampersand replacement=\&, 
    row sep=-1.5ex, column sep=2pt
  }}

\tikzstyle{bulletnode}=[minimum size = 2ex, inner sep=-1pt, draw, circle]
\newcommand{\convexpath}[2]{
  [   
  create hullcoords/.code={
    \global\edef\namelist{#1}
    \foreach [count=\counter] \nodename in \namelist {
      \global\edef\numberofnodes{\counter}
      \coordinate (hullcoord\counter) at (\nodename);
    }
    \coordinate (hullcoord0) at (hullcoord\numberofnodes);
    \pgfmathtruncatemacro\lastnumber{\numberofnodes+1}
    \coordinate (hullcoord\lastnumber) at (hullcoord1);
  },
  create hullcoords
  ]
  ($(hullcoord1)!#2!-90:(hullcoord0)$)
  \foreach [
  evaluate=\currentnode as \previousnode using \currentnode-1,
  evaluate=\currentnode as \nextnode using \currentnode+1
  ] \currentnode in {1,...,\numberofnodes} {
    let \p1 = ($(hullcoord\currentnode) - (hullcoord\previousnode)$),
    \n1 = {atan2(\y1,\x1) + 90},
    \p2 = ($(hullcoord\nextnode) - (hullcoord\currentnode)$),
    \n2 = {atan2(\y2,\x2) + 90},
    \n{delta} = {Mod(\n2-\n1,360) - 360}
    in 
    {arc [start angle=\n1, delta angle=\n{delta}, radius=#2]}
    -- ($(hullcoord\nextnode)!#2!-90:(hullcoord\currentnode)$) 
  }
}

\usepackage{amsmath}
\usepackage[capitalise]{cleveref}

\begin{document}

\title{Fair Correlation Clustering\\ Meets Graph Parameters }

\author{Johannes Blaha\inst{1} \and
Robert Ganian \inst{1} \and
Katharina Gillig\inst{2} \and\\
Jonathan S. Højlev\inst{3} \and
Simon Wietheger\inst{1}}

\authorrunning{Blaha, Ganian, Gillig, Højlev, Wietheger}
\institute{TU Wien, Vienna, Austria \email{e11930322@student.tuwien.ac.at, rganian@gmail.com, swietheger@ac.tuwien.ac.at} \and
Universität Bonn, Bonn, Germany \email{s3kagill@uni-bonn.de} \and 
Technical University of Denmark, Copenhagen, Denmark 
\email{s194684@dtu.dk}}
\maketitle              
\begin{abstract}
We study the generalization of \textsc{Correlation Clustering} which incorporates fairness constraints via the notion of fairlets. The corresponding \textsc{Fair Correlation Clustering} problem has been studied from several perspectives to date, but has so far lacked a detailed analysis from the parameterized complexity paradigm. We close this gap by providing tractability results for the problem under a variety of structural graph parameterizations, including treewidth, treedepth and the vertex cover number; our results lie at the very edge of tractability given the known \NP-hardness of the problem on severely restricted inputs.
\end{abstract}
\section{Introduction}
\label{sec:intro}
Clustering is a fundamental task that arises in a wide range of applications, where the goal is to partition a set of objects based on pairwise similarity or dissimilarity information. On graphs, the task is typically formalized via the 
\textsc{Correlation Clustering} problem~\cite{BansalBC04,Bonchi2022,CohenAddadLN22}: given a graph $G$ and an integer~$k$, can we modify (i.e., add or remove) at most $k$ edges of $G$ to obtain collection of pairwise independent cliques? This problem is sometimes also referred to as \textsc{Cluster Editing}~\cite{FominKPPV14,KomusiewiczU11,KomusiewiczU12}, and is known to be \NP-hard even if we are promised that the optimal solution contains two clusters~\cite{ClusterEditingNPComplete}. The number of modifications is often called the \emph{cost} of a correlation clustering.

While \textsc{Correlation Clustering} can be seen as a well-suited model for finding an optimal clustering of unlabeled elements (i.e., vertices), in many scenarios one needs to ensure that the clustering respects additional constraints arising from the specific context it is applied in. Fairness considerations are a prime example of this: on labeled data, a solution to the base \textsc{Correlation Clustering} problem can yield clusters where the presence of the labels is disproportional when compared to the label distribution on the input. In view of these considerations, Chierichetti, Kumar, Lattanzi and Vassilvitskii initiated the study of fair clustering problems via the lens of \emph{fairlets}~\cite{Chierichetti0LV17}. When viewing the individual vertex labels as colors, a fairlet is a minimum set of vertices which exhibit the same color ratio as the whole graph, and we call a clustering \emph{fair} if each cluster can be partitioned into fairlets\footnote{We remark that this strict formalization of fair clustering can easily be relaxed to allow some ``slack'' in terms of the color ratios. Similarly to previous theoretical studies of the notion, here we focus on the most natural basic variant.}. 
Since its introduction, the complexity of fairlet-based clustering has been studied in a broad range of models~\cite{KleindessnerAM19,AhmadianE0M20,Casel0SW23,BandyapadhyayFS24}; here, we focus on the setting of correlation clustering and formalize our problem of interest below (see also Figure~\ref{fig:FairCorrelationClusteringExample}).

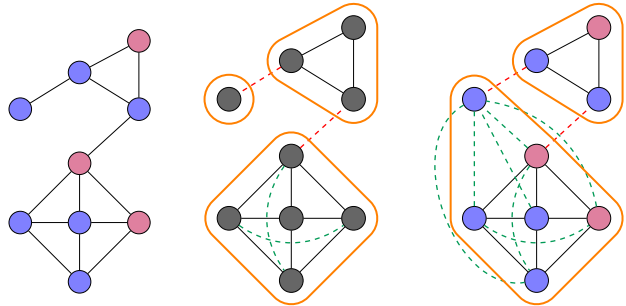
\begin{figure}[th]
    \centering
    \scalebox{0.60}{
    \begin{tikzpicture}[
        blue/.style = {draw, circle, fill=blue!50, minimum size=5mm},
        purple/.style = {draw, circle, fill=purple!50, minimum size=5mm}
        ]
        \node[blue](a){};
        \node[blue, left=8mm of a](b){};
        \node[blue, below=8mm of a](c){};
        \node[purple, right=8mm of a](d){};
        \node[purple, above=8mm of a](e){};
        \node[blue, above=20mm of d](f){};
        \node[blue, above=15mm of e](h){};
        \node[purple, above=10mm of f](g){};
        \node[blue, above=20mm of b](i){};
        \draw[-] (b) edge (a)
                 (c) edge (a)
                 (e) edge (a)
                 (d) edge (a)
                 (b) edge (c)
                 (c) edge (d)
                 (d) edge (e)
                 (e) edge (f)
                 (f) edge (g)
                 (g) edge (h)
                 (h) edge (i)
                 (f) edge (h)
                 (e) edge (b); 
        \draw[white, very thick] \convexpath{i,h,g,d,c,b}{0.5cm};    
    
        \end{tikzpicture}}
        \hskip 0.35cm
    \scalebox{0.63}{
    \begin{tikzpicture}[
        main/.style = {draw, circle, fill=black!60, minimum size=5mm}
        ]
            \node[main](a){};
            \node[main, left=8mm of a](b){};
            \node[main, below=8mm of a](c){};
            \node[main, right=8mm of a](d){};
            \node[main, above=8mm of a](e){};
            \node[main, above=20mm of d](f){};
            \node[main, above=15mm of e](h){};
            \node[main, above=10mm of f](g){};
            \node[main, above=20mm of b](i){};
            \draw[-] (b) edge (a)
                     (c) edge (a)
                     (e) edge (a)
                     (d) edge (a)
                     (b) edge (c)
                     (c) edge (d)
                     (d) edge (e)
                     (f) edge (g)
                     (g) edge (h)
                     (f) edge (h)
                     (e) edge (b);   
            \draw[-,ForestGreen, thick,dashed, bend right=35]   (e) edge (c)
                                                    (b) edge (d);
            \draw[-, red, thick, dashed]    (h) edge (i)
                                            (e) edge (f);                
            \draw[orange, very thick] \convexpath{b,e,d,c}{0.5cm};    
            \draw[orange, very thick] \convexpath{h,g,f}{0.5cm}; 
            \node[draw=orange,very thick,circle,fit=(i)]{};
        \end{tikzpicture}
        }
        \scalebox{0.63}{
        \begin{tikzpicture}[
        blue/.style = {draw, circle, fill=blue!50, minimum size=5mm},
        purple/.style = {draw, circle, fill=purple!50, minimum size=5mm}
        ]
        \node[blue](a){};
        \node[blue, left=8mm of a](b){};
        \node[blue, below=8mm of a](c){};
        \node[purple, right=8mm of a](d){};
        \node[purple, above=8mm of a](e){};
        \node[blue, above=20mm of d](f){};
        \node[blue, above=15mm of e](h){};
        \node[purple, above=10mm of f](g){};
        \node[blue, above=20mm of b](i){};
        \draw[-] (b) edge (a)
                 (c) edge (a)
                 (e) edge (a)
                 (d) edge (a)
                 (b) edge (c)
                 (c) edge (d)
                 (d) edge (e)
                 (f) edge (g)
                 (g) edge (h)
                 (f) edge (h)
                 (e) edge (b);   
        \draw[-, ForestGreen, thick, dashed, bend right=35] 
            (e) edge (c)
            (b) edge (d)
            (d) edge (i);
        \draw[-, ForestGreen, thick,dashed, bend right=80] 
            (i) edge (c);
        \draw[-, ForestGreen, dashed, thick] 
            (a) edge (i)
            (b) edge (i)
            (e) edge (i); 
        \draw[-, red, thick, dashed]
            (h) edge (i)
            (e) edge (f);
        \draw[orange, very thick] \convexpath{b,i,e,d,c}{0.5cm};    
        \draw[orange, very thick] \convexpath{h,g,f}{0.5cm}; 
        \end{tikzpicture}
        }\\
    
    \caption{\emph{Left}: An input graph; here the fairlet is $\{$red, blue, blue$\}$. \emph{Middle:} A ``fairness-oblivious'' minimum-cost correlation clustering (cost 4). \emph{Right}: A minimum-cost fair correlation clustering (cost 9).}
    \label{fig:FairCorrelationClusteringExample}
    \raggedright
    \medskip
\end{figure}

\begin{myproblem}
\problemtitle{\textsc{Fair Correlation Clustering} (FCC)}
\probleminput{An undirected $n$-vertex graph $G$, a budget $B  \in \mathbb{N}$ and a $\kappa$-coloring $\chi:V(G)\rightarrow \{1,\dots,\kappa\}$ for an integer $\kappa \le n$.}
\problemtask{Is there a fair correlation clustering of cost at most $B$?}
\end{myproblem}

As a generalization of \textsc{Correlation Clustering}, \textsc{FCC} is \NP-hard~\cite{AhmadianE0M20}. However, unlike the former problem, \textsc{FCC} was shown to remain \NP-hard even when restricted to very simple input graphs, such as graphs of diameter $2$ and $2$-colored trees of diameter $4$~\cite{Casel0SW23}. 
Most foundational research on \textsc{FCC} to date has focused on overcoming its inherent intractability via approximation. Ahmadian, Epasto, Kumar, and Mahdian~\cite{AhmadianE0M20} reduced the problem to the ``fairlet decomposition'' problem with a standard median cost function to achieve a constant-factor approximation, while Ahmadi, Galhotra, Saha, and Roy~\cite{abs200203508} 
focused on the number of colors and the relative amount of vertices in the most common color and obtained an  approximation, which is quadratic in both parameters.

In this work, we take a different approach and ask whether it is possible to exploit the structural properties of the input graph $G$ to yield exact algorithms for \textsc{FCC}. The \emph{parameterized complexity} paradigm (see Section~\ref{sec:prelims}) offers the perfect tools for answering this question, and we obtain a detailed overview of how the problem's complexity depends on the structure of $G$.

\smallskip
\noindent \textbf{Contributions.}
The aforementioned \NP-hardness of \textsc{FCC} on $2$-colored trees rules out fixed-parameter and even \XP-algorithms under the vast majority of graph parameterizations, including not only the classical \emph{treewidth}~\cite{RobertsonS84} but also the closely related \emph{pathwidth}~\cite{RobertsonS83}, \emph{treedepth}~\cite{sparsity} and even the \emph{feedback edge number}~\cite{BalabanGR24,GanianMOPR24}. Indeed, all of these measure how ``structurally similar'' the input graph is to a tree (or a special class thereof), and all attain small constant values on trees of diameter~$4$. In spite of these lower bounds, as our first contribution we show that one can still achieve fixed-parameter tractability for \textsc{FCC} by using a suitable structural parameter:

\begin{restatable}{theorem}{thmvcn}
\label{thm:vcn}
\fccshort is in \FPT\ w.r.t.\ the vertex cover number of $G$.
\end{restatable}

Here, the \emph{vertex cover number} measures the size of a minimum vertex cover of $G$ and is a baseline parameter that has found a variety of applications on problems where decomposition-based parameters such as treewidth do not yield tractability~\cite{BalkoCGGHVW24,CyganLPPS14,FoucaudGK0IST24,KlemzS24}.
While the structure of graphs with a small vertex cover number might seem simple, establishing Theorem~\ref{thm:vcn} is non-trivial and relies on a combination of structural insights (among others establishing a bound on the cluster sizes), branching techniques and a matching subroutine.

For our next contributions, we revisit the setting of tree-like graphs and in particular the parameterization by treewidth, which can be seen as the most classical and fundamental graph parameter in the field. The aforementioned \NP-hardness reduction of Casel, Friedrich, Schirneck and Wietheger~\cite{Casel0SW23} rules out \XP-tractability w.r.t.\ treewidth even when combined with the number of colors; however, the construction relies on having large fairlets. On the other hand, in many natural scenarios one would expect the fairlet sizes to be small---a fairlet size of $1$ yields the classical \textsc{Correlation Clustering}, while a fairlet size of $2$ captures a balanced clustering with two colors. Hence, it is natural to ask whether one can achieve parameterized tractability for \textsc{FCC} when combining treewidth with restrictions on the fairlet size. We answer this affirmatively by establishing:

\begin{restatable}{theorem}{thmxptw}
\label{thm:xptw}
\fccshort is in $\XP$ w.r.t.\ the treewidth of $G$ plus the fairlet size.
\end{restatable}

\begin{restatable}{theorem}{thmfpttw}
\label{thm:fpttw}
\fccshort is in \FPT\ w.r.t.\ the treewidth of $G$ when restricted to instances with fairlet size at most $2$.
\end{restatable}

We remark that in spite of the similarity of these two results, they are in fact incomparable and neither implies the other. Theorem~\ref{thm:fpttw} is interesting already for the case of fairlet size $1$, as the fixed-parameter tractability of \textsc{Correlation Clustering} (i.e., \textsc{Cluster Editing}) w.r.t.\ treewidth was not known and previously posed as an open question~\cite{Casel0SW23}. That being said, the main technical challenges arise when dealing with fairlet size $2$. To prove the theorem, we first establish that there exists a solution such that each cluster of size greater than~$2$ forms a connected subgraph in $G$. This structural insight then enables the use of treewidth-based dynamic programming, where connectivity is used to handle ``complex'' clusters while simple but potentially disconnected clusters are tracked in an aggregated manner.

For Theorem~\ref{thm:xptw}, the above approach unfortunately cannot be replicated---it is not difficult to construct instances where a solution must include clusters which are not connected in the original graph. This in turn requires the use of a more careful dynamic programming procedure which tracks information about sets of ``incomplete'' clusters, and the associated dynamic programming tables then require \XP-space w.r.t.\ both the fairlet size and treewidth.

Theorems~\ref{thm:xptw}~and~\ref{thm:fpttw} substantially expand the frontiers of tractability for \textsc{FCC} and the former generalizes the known \XP-algorithm for \textsc{FCC} w.r.t.\ the fairlet size on trees~\cite{Casel0SW23}. However, they also give rise to a natural follow-up question: is the problem fixed-parameter tractable w.r.t.\ treewidth plus the fairlet size? None of the known hardness results rules out such an outcome, but---as discussed above---the insights used to establish the former do not generalize to the setting with fairlet size $3$. While we leave this as a prominent open question for future work, as our final contribution we provide an affirmative answer when treewidth is replaced with \emph{treedepth}:

\begin{restatable}{theorem}{thmtd}
\label{thm:td}
\fccshort is in \FPT\ w.r.t.\ the treedepth of $G$ plus the fairlet size.
\end{restatable}

Treedepth is a parameter which measures distance to trees of bounded height; it has fundamental connections to the theory of graph sparsity~\cite{sparsity} and has found a variety of algorithmic applications as well~\cite{EibenGKOPW19,NederlofPSW23,BodlaenderGP23}. While Theorem~\ref{thm:fpttw} can be seen as a step towards settling the fixed-parameter tractability of \textsc{FCC} w.r.t.\ treewidth plus fairlet size, the technique underlying its proof is entirely different from the dynamic programming based approach used in Theorems~\ref{thm:xptw}~and~\ref{thm:fpttw}. In particular, we design a reduction rule which operates on a treedepth-based decomposition of the graph and whose exhaustive application results in a graph where each connected component has size bounded by the two parameters. For ``typical'' problems, such a reduction rule would immediately imply fixed-parameter tractability via \emph{kernelization}~\cite{CyganFKLMPPS15}---however, for \textsc{FCC} one cannot treat each connected component separately. Hence, after applying the aforementioned preprocessing the algorithm uses the bound on the size of connected components to solve the problem via an encoding as an Integer Linear Program where the number of variables and constraints is upper-bounded by a function of our parameters; at this point, one can invoke known techniques for ILP solving to obtain the desired result.

\section{Preliminaries}
\label{sec:prelims}
We assume familiarity with standard graph terminology\ifshort
~and the basic notions in parameterized complexity~\cite{DowneyF13,CyganFKLMPPS15} ($\spadesuit$)\fi. 
For $n\in \N$, set $[n] = \set{1, \ldots, n}$.
For a vertex $v$, let $\delta(v)$ denote the degree of $v$ in $G$ and let $\delta_S(v)$ be the number of neighbors of $v$ in a vertex set $S\subseteq V(G)$. 

Consider a vertex coloring $\chi:V(G)\rightarrow [\kappa]$ for a set $[\kappa]$ of \emph{colors}, with $\kappa \in [n]$.
Recalling the discussion in Section~\ref{sec:intro}, a \emph{fairlet} $F$ is a minimum multiset of colors from $[\kappa]$ such that the multiset $\{\chi(v)~|~v\in V(G)\}$ of colors occurring in $G$ can be partitioned into copies of $F$. Note that the fairlet $F$ can be computed from an instance of \textsc{FCC} in polynomial time, and we will denote the fairlet size $|F|$ as $\tilde{c}$; we will sometimes represent fairlets as a vector $(c_1,\dots,c_\kappa)$ encoding the number of times each of the respective colors occurs in $F$. 
A vertex subset $W$ is \emph{fair} if $\{\chi(v)~|~v\in W\}$ can be partitioned into copies of $F$. A \emph{clustering} is a partition of $V(G)$ where each part of the clustering is called a \emph{cluster}. The \emph{cost} of a clustering is the total number of edges in $E(G)$ with endpoints in different clusters plus the total number of non-edges inside each cluster, and a clustering is \emph{fair} if each of its clusters is fair. In view of the definition of \textsc{FCC}, we call a clustering a \emph{solution} if it is fair and achieves a minimum cost. We remark that each of the algorithms presented in this paper can also output a solution as a witness for a \yesinstance.
\iflong

\smallskip
\noindent \textbf{Parameterized Complexity.}
In parameterized complexity~\cite{DowneyF13,CyganFKLMPPS15}, the complexity of a problem is studied not only with respect to the input size, but also with respect to some problem parameter(s). 
The core idea behind parameterized complexity is that the combinatorial explosion resulting from the \NP-hardness of a problem can sometimes be confined to certain structural parameters that are small in practical settings. 

Formally, a {\it parameterized problem} $Q$ is a subset of $\Omega^* \times \N$, where $\Omega$ is a fixed alphabet. 
Each instance of $Q$ is a pair $(\mathcal{I}, k)$, where $k \in \N$ is called the {\it parameter}. 
The class \FPT \space denotes the class of all fixed-parameter tractable parameterized problems, that is, problems which can be solved in time $f(k)\cdot |\mathcal{I}|^{\bigoh(1)}$ for some computable function $f$. A weaker tractability class is \XP, which contains all parameterized problems that can be solved in time $|\mathcal{I}|^{f(k)}$ for some computable function $f$.
\fi

\smallskip
\noindent \textbf{Parameters of Interest.} The \emph{vertex cover number} of a graph $G$ ($\vcn(G)$) is the size of its minimum vertex cover. Recall that a vertex cover is a subset $X$ of vertices such that for each edge $e$ in the graph, at least one endpoint of $e$ lies in $X$. We note that the vertex cover number---along with a minimum-size vertex cover as a witness---can be computed in time $(1.25284^{\vcn(G)})\cdot |V(G)|^{\bigoh(1)}$~\cite{VC_FPT}.

A \emph{nice tree decomposition} of an undirected graph $G$ is a pair $(T, \mathcal{B})$, where $T$ is a tree (whose vertices are called \emph{nodes}) rooted at a node $r$ and $\mathcal{B}$ is a set $\{B_t\mid t\in T\}$ where $\forall t,B_t\subseteq V(G)$ such that the following hold:
\begin{itemize}
	\item For every $\{u,v\} \in E(G)$, there is a node $t$ such that $u,v \in B_t$.
	\item For every vertex $v \in V(G)$, the set of nodes $t$ satisfying $v \in B_t$ forms a subtree of $T$.
	\item $|B_\ell| = 1$ for every leaf $\ell$ of $T$ and $|B_r| = 0$.
	\item There are only three kinds of non-leaf nodes in $T$:
	\begin{itemize}
		\item \textsf{introduce}: a node $t$ with exactly one child $t'$ such that $B_t = B_{t'} \cup \{v\}$ for some $v \notin B_{t'}$.
		\item \textsf{forget}: a node $t$ with exactly one child $t'$ such that $B_t = B_{t'} \setminus \{v\}$ for some $v \in B_{t'}$.
		\item \textsf{join}: a node $t$ with exactly two children $p, q$ such that $B_t = B_p=B_q$.
	\end{itemize}
\end{itemize}

We call each set $B_t$ a \emph{bag}. The width of a nice tree decomposition $(T, \mathcal{B})$ is the size of the largest bag $B_t$ minus 1, and the \emph{treewidth} of $G$ ($\tw(G)$) is the minimum width of a nice tree decomposition of $G$.
A nice minimum-width tree decomposition of an undirected graph $G$ can be computed in time $|V(G)|^{\bigoh(1)}\cdot \tw(G)^{\bigoh(\tw(G)^3)}$~\cite{Bodlaender_1996}; more efficient approximation algorithms are known as well~\cite{Korhonen21}. 

Finally, we formalize a few notions needed to define treedepth.
A \emph{rooted forest} $\mathcal F$ is a disjoint union of rooted trees.
For a vertex~$x$ in a tree~$T$ of $\mathcal F$, the \emph{height} (or {\em depth}) of~$x$ in $\mathcal F$ is the number of vertices on the path from the root of~$T$ to~$x$.
The \emph{height of a rooted forest} is the maximum height of a vertex of the forest. 
Let $V(T)$ be the vertex set of any tree $T \in \mathcal{F}$.
Let the \emph{closure} of a rooted forest~$\mathcal{F}$ be the graph $\clos({\mathcal{F}})$ with the vertex set $V(\clos({\mathcal{F}}))=\bigcup_{T \in \mathcal{F}} V(T)$ and the edge set $E(\clos({\mathcal{F}}))=\{xy \mid \text{$x$ is an ancestor of $y$ in some $T\in\mathcal{F}$}\}$.
A \emph{treedepth decomposition} of a graph $G$ is a rooted forest $\mathcal{F}$ such that $G \subseteq \clos(\mathcal{F})$.
The \emph{treedepth} of a graph~$G$ ($\td(G)$) is the minimum height of any treedepth decomposition of $G$.
An optimal treedepth decomposition can be computed in time $2^{\td(G)^2}\cdot |V(G)|$~\cite{ReidlRVS14}.

It is known (and easy to verify) that $\vcn(G)\geq \tw(G)$ and $\td(G)\geq \tw(G)$.


\section{A Note on Fair Clusters in Tree-like Graphs}

Before proceeding to our main results, we prove a useful lemma that bounds the size of clusters occurring in a solution to an instance of \textsc{FCC} from above in terms of the treewidth of the graph. \ifshort The central idea is that graphs of small treewidth do not have too many edges, so splitting a larger cluster into two will cut fewer edges than it resolves non-edges, which then no longer lie within a cluster. \fi
Since the treewidth bounds both the vertex cover number and the treedepth from below, the same bound can be carried over also to the latter two parameters.

\iflong
\begin{lemma}
\fi
\ifshort
\begin{lemma}[$\spadesuit$]
\fi
\label{lem:MaximumClustersizeTreewidth}
    Every \yesinstance of \fccshort with treewidth $\tw$ and fairlet size~$\tilde{c}$ admits a solution such that each cluster has size at most $\max(24\tw,\tilde{c})$.
\end{lemma}

\iflong
\begin{proof}
Towards a contradiction, assume that a given instance admits a solution~$\mathcal{P}$ containing a cluster $Z$ such that $|Z|>\max(24\tw,\tilde{c})$. Since $Z$ is fair, it contains precisely $d\cdot \tc$ vertices for some integer $d$, and by assumption we have $d\geq 2$. Consider an arbitrary partition of $Z$ into fair clusters $Z_1$ and $Z_2$ such that $|Z_1|=\lfloor \frac{d}{2} \rfloor \cdot \tc$ and $|Z_2|=\lceil \frac{d}{2} \rceil \cdot \tc$. Let $\mathcal{P}'$ be the fair clustering obtained from~$S$ by replacing the cluster $Z$ with the two clusters $Z_1$, $Z_2$.

Let $\cost(\mathcal{P})$ and $\cost(\mathcal{P}')$ denote the costs of the respective solution. Observe that by construction, the latter can be obtained from the former by accounting for the edges and non-edges between $Z_1$ and $Z_2$, as follows: 
\begin{align*}
\cost(\mathcal{P}')=\cost(\mathcal{P}) & +  |\{\{v,w\}~|~v\in Z_1, w\in Z_2, \set{v,w}\in E(G)\} | \\
 & - |\{\{v,w\}~|~v\in Z_1, w\in Z_2, \set{v,w}\not \in E(G)\} |.
\end{align*}

It is known (and easy to observe from the definition of tree-decompositions) that every graph of treewidth $\tw$ has at most $\tw\cdot |V(G)|$ edges, and that the class of graphs of bounded treewidth is hereditary. Since the graph induced on $Z$ has treewidth at most $\tw$, the set on the first row in the above equation has size at most $\tw\cdot |Z|$, while the set in the second row has size at least ${|Z_1|\cdot |Z_2|}-(\tw\cdot |Z|)$. We now consider the following two subcases depending on the value of $d$.

First, assume $d\geq 3$. Then we obtain 
$\tw\cdot |Z| < {|Z_1| \cdot |Z_2|}-(\tw\cdot |Z|)$; indeed, by using 
$|Z| = \tc d >24\tw$ we have:
\begin{align*}
    |Z_1|\cdot |Z_2| 
    &\ge \tc(\lceil\tfrac{d}{2}\rceil-1) \cdot \tc\lceil\tfrac{d}{2}\rceil 
    \ge \tc^2d^2(\tfrac{1}{4}-\tfrac{1}{2d})
    \ge \tfrac{1}{12} \tc^2d^2 > 2\tw \cdot |Z|.
\end{align*}

Second, assume $d=2$. Then, since $|Z| = 2\tc >24\tw$  we obtain an even stronger inequality---indeed:
\begin{align*}
    |Z_1|\cdot |Z_2| 
    =\tc^2 & > 12\tw \cdot \tc
     = 6\tw \cdot |Z|.
\end{align*}

Thus, $\cost(\mathcal{P}') < \cost(\mathcal{P})$, contradicting $\mathcal{P}$ being a solution. We remark that this bound suffices for our purposes, but that we do not claim it to be tight.
\qed
\end{proof} 
\fi

\section{Tractability via Vertex Cover Number}

In this section, we establish our first tractability result, which is also the only one that does not require the parameterization to include the fairlet size:

\thmvcn*

Since all pairs of vertices inside a cluster consisting of a set of pairwise independent vertices contribute to the cost, we have:
\begin{observation}\label{obs:vertexcover_remaining_easy}
    Let $(G,B,\chi)$ be an instance of $\fcc$ with a solution $\mathcal{P}$ and let $M$ be a vertex cover of $G$. Then every cluster $C$ in $\mathcal{P}$ such that $C\cap M=\emptyset$ satisfies $|C|=\tilde{c}$.
\end{observation}

The above observation implies that given a subset $Z\subseteq V(G)\setminus M$ of vertices outside the vertex cover, a minimum-cost fair clustering of $Z$ (if one exists) will ``incur'' a cost of precisely $\binom{\tc}{2}\cdot \frac{|Z|}{\tc}$ plus the number of edges incident to $Z$. Crucially, this does not depend on the specific partitioning of  $Z$ into fair clusters.

Recalling Lemma~\ref{lem:MaximumClustersizeTreewidth} and that $\vcn(G)\geq \tw(G)$, we notice that there are at most $\left\lceil\frac{24k}{\tc}\right\rceil$ different possible sizes a cluster in a solution can have. Indeed, every cluster in a solution must have a size from the set $\{\tc d \mid d\in [\left\lceil\frac{24}{\tc}\right\rceil]\}$.  
As a second prerequisite of our proof we establish the following reformulation of the correlation clustering cost which shows that, when the sizes of all clusters are prescribed, minimizing the cost is equivalent to maximizing the number of edges within clusters. 

\begin{observation}\label{obs:maximize_uncut_edges}
    Consider a graph $G$ and a clustering $\mathcal{P}$ of $V(G)$. Let $a$ be the number of unordered pairs of vertices $\set{v,w}$ such that $v$ and $w$ share a cluster in $\mathcal{P}$ and let $b$ be the number of edges $\set{v,w}$ such that $v$ and $w$ share a cluster in $\mathcal{P}$. Then, the cost of $\mathcal{P}$ is $|E(G)| + a - 2b$ since the number of non-edges inside clusters is $(a-b)$ and the number of edges between clusters is $(|E(G)|-b)$.
\end{observation}

We now have all the ingredients required to prove our first result. \ifshort($\spadesuit$)\fi

\ifshort
\begin{proof}[Proof Sketch for Theorem~\ref{thm:vcn} $(\spadesuit)$]
Given an instance $(G,B,\chi)$ of \textsc{FCC}, we first compute a minimum vertex cover $M$ of size $k$ using known algorithms~\cite{VC_FPT}. We apply exhaustive branching to determine which of the vertices in $M$ will belong to the same cluster in a hypothetical solution, and also to determine the exact size of each cluster intersecting $M$. By following up on the above arguments, we can bound the number of branches that need to be considered by $k^k\cdot (24k)^k$.

In each branch, we obtain a \emph{pre-clustering} which provides partial information about the clusters intersecting $M$. With Observation~\ref{obs:maximize_uncut_edges}, we can reduce the task of completing the clusters intersecting $M$ to the problem of finding a maximum-weight matching in an auxiliary bipartite graph, where one side corresponds to the ``free spots'' in the pre-clustering and the other side corresponds to vertices outside of $M$. An illustration of this step is provided in Figure~\ref{fig:MaximumWeightMatching}. 

Finally, the remaining vertices are assigned to fair clusters of size $\tc$ in an arbitrary fashion, whereas correctness is argued via Observation~\ref{obs:vertexcover_remaining_easy}. 
\qed
\end{proof}
\fi

\iflong
\begin{proof}[Proof of Theorem~\ref{thm:vcn}]
Let $(G,B,\chi)$ be an FCC instance with $n=|V(G)|$, $\kappa$ colors and fairlet vector $(c_1,\dots,c_\kappa)$. Firstly, a minimum-sized vertex cover $M$ is computed in time $\bigoh^*(1.25284^k)$, where $k$ is this minimum size. We determine the clusters that intersect $M$ and then arbitrarily assign the rest of the graph to fairlet-sized clusters. 

We start by considering all possible partitions of the vertex cover $M$ and for each of these partitions branch on (guess) which final size the cluster containing these sets of vertex cover vertices should have. Due to Lemma~\ref{lem:MaximumClustersizeTreewidth}, there are at most $24k$ options for the size of any cluster, and we only consider branches in which the guessed sizes sum up to at most $n$.

We refer to a partition of $M$ with guessed cluster sizes as a \emph{pre-clustering}, and to each of its clusters as a \emph{pre-cluster}.
Note that for each pre-cluster, we can easily derive from its composition and assumed size how many vertices of each color are still missing. Then the problem of assigning the non-vertex-cover vertices to the pre-clusters becomes a matching problem between the available ``spots'' and the rest of the graph. To find the optimal such assignment w.r.t. the cost, we can construct a special bipartite graph as follows and compute a maximum-weight matching, see Figure~\ref{fig:MaximumWeightMatching}. 

The vertices on the left side of the bipartite graph are the at most $n$ spots in the pre-clusters containing some vertex cover vertices. On the right are all remaining vertices $V(G)\setminus M$. For each spot $s$ of color $i$, construct an edge $e$ with weight $w(e)$ to all $v\in V(G)\setminus M$ with the same $i$. The weight $w(e)$ is defined as the number of neighbors vertex $v$ has in the pre-cluster of $s$. 
Note that the pre-clustering predetermines the number of clusters of each size in any clustering to which it is a pre-clustering. As this predetermines the value $a$ in Observation~\ref{obs:maximize_uncut_edges}, it remains to minimize the total cost of the clustering by maximizing the number edges within clusters.
Thus, computing a maximum weight matching will yield an optimal assignment of vertices for that specific pre-clustering. 
Lastly, due to Observation~\ref{obs:vertexcover_remaining_easy}, we can now arbitrarily cluster the remaining vertices into fairlets. Once a final solution is computed, it is easily verified, and if the cost is less than $B$, accept the instance. 

If the instance is accepted, then we indeed found a solution since it is checked explicitly. Conversely, if a solution of cost at most $B$ exists, it induces a partition of $M$ and our algorithm will certainly consider the same partition (since it considers all partitions) with the same sizes of those clusters. As the algorithm computes a minimum-cost way of extending this pre-clustering to a full clustering, it finds a solution of cost at most $B$ and accepts the instance. Overall, this means an optimal solution will not be overlooked by our algorithm. 

The total running time is FPT w.r.t.\ $k$. There are at most $k^k$ ways to partition $M$, and at most $24k$ choices for the size of each of the at most $k$ pre-clusters. We thus consider at most $((24k)^2)^k = k^{\bigoh(k)}$ distinct pre-clusterings. Constructing the bipartite graph for the maximum-weight matching requires $\bigoh(kn^2)$ time, as there are $\bigoh(n^2)$ edges overall and for each edge, we can compute its weight by iterating over the vertex cover nodes in the cluster in $\bigoh(k)$ time. As we have integer weights, computing a maximum weight matching on a bipartite graph $B$ can be solved in $\bigoh(\sqrt{n}n^2 \log(n))$ time with the approach by Gabow and Tarjan~\cite{matchingruntime}. 

Finally, clustering the remaining vertices outside the existing clusters takes an additional $\bigoh(n)$ time. For that we can first bucket the vertices by color in $\bigoh(n)$ time, and then fill up the fairlets by choosing the vertices from the buckets according to the fairlet vector.
This yields a total runtime of 
\[
    \bigoh\left(k^{\bigoh(k)}\left(kn^2+\sqrt{n}n^2\log(n^2)+n\right)\right) = k^{\bigoh(k)} \cdot \bigoh\left(\sqrt{n}n^2\log n\right).  \eqno\qed
\]
\end{proof}
\fi

\begin{figure}[t]
    \centering
    \resizebox{0.5\textwidth}{!}{
    \begin{tikzpicture}[
        blue/.style = {draw, circle, fill=blue!50, minimum size=4mm},
        blueplace/.style = {draw=blue, thick, circle, dashed, fill=blue!10, minimum size=4mm},
        purple/.style = {draw, circle, fill=purple!50, minimum size=4mm},
        purpleplace/.style = {draw=purple, thick,dashed, circle, fill=purple!10, minimum size=4mm},
        green/.style = {draw, circle, fill=green!50, minimum size=4mm},
        greenplace/.style = {draw=green, thick,dashed, circle, fill=green!10, minimum size=4mm}
        ]
        \node[blue](a){};
        \node[blue,     below=6mm of a](b){};
        \node[blue,     below=6mm of b](c){};
        \node[green,    below=6mm of c](d){};
        \node[green,    below=6mm of d](e){};
        \node[purple,   below=6mm of e](f){};
        \node[blueplace,    left=28mm of a](la){};
        \node[purpleplace,  left=28mm of b](lb){};
        \node[purpleplace,  left=28mm of c](lc){};
        \node[blueplace,    left=28mm of d](ld){};
        \node[greenplace,   left=28mm of e](le){};
        \node[greenplace,   left=28mm of f](lf){};
        \node[left=20mm of lb](vcA){};
        \node[green, double, very thick, above=1.5mm of vcA](vca){};
        \node[green, double, very thick, below=1.5mm of vcA](vcb){};

        \node[left=20mm of le](vcB){};
        \node[purple, double, very thick, above=1.5mm of vcB](vcc){};
        \node[purple, double, very thick, below=1.5mm of vcB](vcd){};

        \node[ above=6mm of la, xshift=4mm]{\Large{pre-clusters\hspace{6mm}available spots\hspace{6mm}rest of graph }};

        \node[left=12mm of a](w){};
        \node[above=0mm of w]{\large{w(e)}};
        
        \draw[-]
            (la) edge (a)
            (la) edge (b)
            (la) edge (c)
            (lb) edge (f)
            (lc) edge (f)
            (ld) edge (a)
            (ld) edge (b)
            (ld) edge (c)
            (le) edge (d)
            (le) edge (e)
            (lf) edge (d)
            (lf) edge (e)
            ;
        \draw[orange, very thick] \convexpath{vcb, vca, la, lc}{4mm};
        \draw[orange, very thick] \convexpath{vcd, vcc, ld, lf}{4mm};
    \end{tikzpicture}
    }
    \caption{The auxiliary graph used in the proof of Theorem~\ref{thm:vcn}; the available spots are matched with non-vertex cover nodes of the same color.}
    \label{fig:MaximumWeightMatching}
\end{figure}
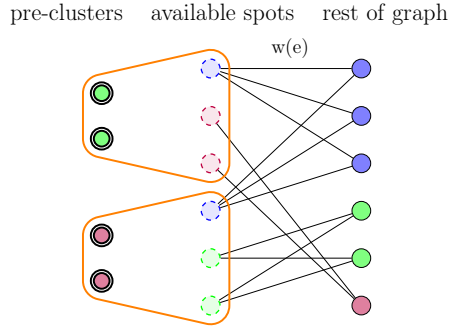

\section{XP Tractability via Treewidth and Fairlet Size}
The goal of this section is to provide an \XP-algorithm for \textsc{Fair Correlation Clustering} parameterized by $\tw(G)+\tc$. Towards this, we introduce some additional terminology that will be useful in our dynamic-programming based proof.

Given a node $t$ in a tree decomposition $T$, let $G_t\subseteq G$ be the subgraph of $G$ that is induced by the vertices in the bags of the subtree rooted at $t$. We refer to the vertices in $V(G_t)\setminus B_t$ as \emph{past vertices}, and to the subgraph of $G$ induced on the past vertices as the \emph{past graph}.
Similarly, $V(G)\setminus V(G_t)$ is the set of \emph{future vertices} and these induce the \emph{future graph}. 

\thmxptw*
\iflong
\begin{proof}
\fi
\ifshort
\begin{proof}[Proof Sketch $(\spadesuit)$]
\fi
Let $(G, B, \chi)$ be an instance of \textsc{FCC} with $n=|V(G)|$ and fairlet vector $(c_1,\dots,c_\kappa)$. Let $(T,\mathcal{B})$ be a nice tree decomposition of $G$ of width at most $w \coloneqq 5\tw+4$. Such a decomposition consisting of $\bigoh(n)$ nodes can be computed in time $n\cdot 2^{\bigoh(w)}$~\cite{Bodlaender_1996}. We process the nodes of $T$ in a bottom-up fashion with a dynamic programming algorithm Å. At each node $t$ of $T$ with an associated bag $B_t$, Å computes a table $\mathcal{M}_t$ of \textit{records}, which are tuples of the form $(\open,\current,\cost)$, each representing a possible partition of the vertices of the graph $G_t$ within a hypothetical solution to the \textsc{FCC} instance. For this, we define:
\begin{itemize}
    \item $\open$ is a multiset of tuples of the form $(\size,X)$. There is one tuple for each set in the partition which is disjoint from the bag $B_t$ and a proper subset of a cluster in the hypothetical solution.
    Here, $\size$ is the predicted size of the cluster in the hypothetical solution, which we write as a multiple of the fairlet size, $\size = d\cdot \tc$, and $X$ is a tuple of the form $X=(x_1,\dots,x_\kappa)$, where $x_i$ is the number of vertices of color $i$ in the set, for each color $i$. 
    \item $\current$ is a set of tuples of the form $(\size,X,S)$. There is one tuple for each set in the partition intersecting $B_t$. We let the $\size$ be as above, $S\subseteq B_t$ denote the intersection with the bag $B_t$ and $X$ similarly as above denotes how many vertices of each color there are in the set when \emph{not counting vertices in $S$}. Let $S_i$ denote the number of vertices of color $i$ in $S$.
    \item $\cost$ is an integer representing the cost of the clustering so far. More precisely, it is the cost in the hypothetical solution induced by all (non)-edges in $G$ incident to at least one vertex in the past.
\end{itemize}

Algorithm Å operates as follows. It starts by computing a table of records $\mathcal{M}_t$ for each leaf $t$ of $T$. It then proceeds computing the records in a bottom-up way until it computes the table for the root $\mathcal{M}_r$. Then, $(G,B,\chi)$ is a \yesinstance if and only if there exists a tuple $(\emptyset,\emptyset,B')\in \mathcal{M}_r$ with $B'\leq B$. 
The procedures for each of the four different types of tree nodes $t$ are as follows.

\paragraph*{\textsf{Leaf} node.} Let $v$ be the single vertex in the bag $B_t$. Add one record for each possible final size of the cluster containing $v$, all with $X=(0,\dots,0)$. As there are no past vertices yet, the cost is $0$ and $\open = \emptyset$. So let
\begin{align*}
   \mathcal{M}_t=\left\{\Big(\emptyset, \set{C_d},0\Big)\,\middle|\,d\in \left[\left\lceil\frac{24\tw}{\tc}\right\rceil\right]\right\}  \hspace{2mm}\text{with}\hspace{2mm} C_d = \big(\tc d,(0, \ldots, 0),\{v\}\big).
\end{align*}

\paragraph*{\textsf{Introduce} node.} Let $p$ be the child of $t$ in $T$, and let $v$ be the vertex introduced at $t$, i.e. $\set{v}= B_t\setminus B_p$. For each $(\open,\current,\cost)\in \mathcal{M}_p$, Å branches on where to put $v$. We consider all of the following options (with all choices of a current or open cluster in the latter two cases) and add the respective records to $\mathcal{M}_t$:
\begin{enumerate}
    \item $v$ starts a new cluster.
    \begin{itemize}
        \item This is very similar to what happens in a leaf node. Add the record $(\open,\current\cup \set{C_d},\cost)$ to $\mathcal{M}_t$ for each $d\in \left[\left\lceil\frac{24\tw}{\tc}\right\rceil\right]$, where $C_d$ is as above.
    \end{itemize}
    \item $v$ joins a current cluster $C=(\tc d,(x_1,\dots,x_\kappa),S)$ s.t.\ $S_{\chi(v)}+x_{\chi(v)} < c_{\chi(v)} d$.
    \begin{itemize}
        \item Let $\current' = \current \cup \set{(\tc d,(x_1,\dots,x_\kappa),S\cup\set{v})} \setminus \set{C}$ and add the record $(\open,\current',\cost)$ to $\mathcal{M}_t$.
    \end{itemize}
    \item $v$ joins an open cluster $C=(\tc d,(x_1,\dots,x_\kappa))$ s.t.\ $x_{\chi(v)} < c_{\chi(v)} d$.
    \begin{itemize}
        \item In this case, a past cluster becomes a current cluster. Add the record $(\open',\current',\cost)$ to $\mathcal{M}_t$, where $\open' = \open \setminus \set{C}$, and $\current' = \current \cup \set{(\tc d, (x_1,\dots,x_\kappa), \set{v})}$.
    \end{itemize}
\end{enumerate}
In each case, the cost stays the same as the set of past vertices does not change.

\paragraph*{\textsf{Forget} node.} Let $p$ be the child of $t$ in $T$, and let $v$ be the vertex forgotten at $t$, i.e. $\set{v}= B_p\setminus B_t$. 
For each $(\open,\current,\cost)\in \mathcal{M}_p$ we add one record $(\open',\current',\cost')$ to $\mathcal{M}_t$ as follows.
Consider the entry $C = (\tc d, (x_1, \ldots, x_\kappa), S)\in \current$ such that $v\in S$.
As $v$ becomes a past vertex, we have to update the cost. Its (non)-edges to other past vertices are already accounted for in $\cost$, so we obtain the new $\cost'$ by increasing $\cost$
\begin{itemize}
    \item  by $1$ for each non-edge to other vertices in $S$,
    \item  by $1$ for each edge to vertices in $B_t\setminus S$,
    \item by $\tilde{c}d - |S| - \sum_{i\in [\kappa]} x_i$ to account for the non-edges to all future vertices added to the cluster.
\end{itemize}

If the cluster is full, i.e., $S_i + x_i = c_i d$ for each color $i$, we do not track it any longer. 
In this case, $\current' = \current \setminus \set{C}$ and $\open' = \open$.
Else, if $v$ was the last current vertex of its cluster ($S=\set{v}$) then the cluster is moved from $\current$ to $\open$.
In this case, $\current' = \current \setminus \set{C}$ and $\open' = \open \cup \set{(\tc d, (x_1, \ldots, x_{\chi(v)}+1, \ldots, x_\kappa))}$.
Otherwise, $v$ simply removed from $S$ in its current cluster: $\open' = \open$ and 
\[\current' = \current \cup \set{(\tc d, (x_1, \ldots, x_{\chi(v)}+1, \ldots, x_\kappa), S\setminus\set{v})} \setminus \set{C}.\]

\paragraph*{\textsf{Join} node.} Let $p$ and $q$ be the two children of $t$ in $T$, and recall that $B_p=B_q=B_t$. Conceptually, the nodes $p,q$ have the same current vertices but disjoint pasts. 
For each combination of some $(\open_p,\current_p,\cost_p)\in \mathcal{M}_p$ and $(\open_q,\current_q,\cost_q)\in \mathcal{M}_q$, we check if the two records are \emph{compatible}, and if so, join them to add a number of new records $(\open',\current',\cost')$ to $\mathcal{M}_t$. 

The two records are compatible if they agree on the shape of the clusters intersecting $B_t$. Formally, this means that for each $(\tc d,(p_1,\dots,p_\kappa),S)\in \current_p$ there is $(\tc d,(q_1,\dots,q_\kappa),S)\in \current_q$ (and vice versa), and for each such pair and each color $i\in [\kappa]$ we have $p_i + q_i + S_i \le c_i d$. Define $\current'$ as the set of entries $(\tc d, (p_1 + q_1, \ldots, p_\kappa+q_\kappa), S)$ of all these pairs.

Each compatible pair now adds (one or more) records $(\open', \current', \cost')$ to $\mathcal{M}_t$ as follows. While $\current'$ is uniquely determined as described above, there are multiple options for $\open'$ as a cluster in a hypothetical solution might contain vertices from both $\open_p$ and $\open_q$. 
Thus, consider each pair of a distinct entry in $\open_p$ and a distinct entry in $\open_q$ and branch on how many of the entries of the first kind are merged with an entry of the second kind. Only consider pairings $(\tc d_p, (p_1,\ldots,p_\kappa))$ and $(\tc d_q, (q_1,\ldots,q_\kappa))$ where both entries have the same assumed final size, that is, $d_p = d_q$. They are merged into an entry $(\tc d_p, (p_1+q_1, \ldots, p_\kappa+q_\kappa))$ and any branch with a pair where $p_i+q_i > c_i d_p$ for some color $i$ is disregarded. Further disregard branches that build more pairs than there are entries of the respective types in the respective pasts.
For each remaining branch, let $\open'$ correspond to the multiset of open cluster types after the respective merging (where entries that are not merged are kept the same way as they were in $\open$).

Last, we determine $\cost'$. Let $P$ and $Q$ denote the sets of vertices in the past of nodes $p$ and $q$, respectively. Then the set of vertices in the past of $t$ is $P\cup Q$, where we remark that $P$ and $Q$ are disjoint by the definition of tree decompositions. Hence, we can compute $\cost'$ as the sum of $\cost_p$ and $\cost_q$ while subtracting the cost incurred by (non)-edges between $P$ and $Q$, as they would be accounted for twice, otherwise. By the definition of tree decomposition, there are no edges connecting $P$ and $Q$. Non-edges only incur a cost if they are incident to two vertices in the same cluster, so only non-edges in newly merged clusters have to be considered. Indeed, for each newly merged cluster obtained from entries $(\tc d_p, (p_1,\ldots,p_\kappa))$ and $(\tc d_q, (q_1,\ldots,q_\kappa))$, the number of non-edges counted twice is $\left(\sum_{i\in[\kappa]} p_i\right) \cdot \left(\sum_{i\in[\kappa]} q_i\right)$. Subtracting the total number of these non-edges among all newly merged clusters from $\cost_p+\cost_q$, yields $\cost'$.

\ifshort
To complete the proof, it remains to argue correctness and that the algorithm terminates in time at most $n^{\bigoh(\tw+\tc)^{\kappa+1})}$.
\fi
\iflong
\paragraph*{Correctness and Required Time.}
Suppose $(G,B,\chi)$ is a \yesinstance witnessed by some fair clustering $\mathcal{P}$. 
Note that at each node $t$, there is a record describing the subpartition induced by the vertices of $G_t$ on $\mathcal{P}$ and the algorithm identifies this record --- as it can be constructed from (at least one) subpartitioning (that is, a record) from a child node of $t$ (and is trivial for leafs). Recall that at the root node, all vertices are in the past. Hence, in a record for $\mathcal{P}$ at the root, $\cost$ precisely describes the cost of $\mathcal{P}$ and $\open$ and $\current$ are empty. Hence, the algorithm accepts the instance.

For the other direction, assume that the algorithm would accept the instance due to a record $(\emptyset, \emptyset, B')$ with $B'\le B$. Our construction ensures that for each record $(\open, \current, \cost)$ in a table at some node $t$, there is a partition of the vertices in $G_t$ that matches the description of $\open$,  $\current$, and $\cost$ as described above (for some assigned ``final'' $\size$ for each set). Consider a partition $\mathcal{P}$ for the record $(\emptyset, \emptyset, B')$ at the root node. Then $\mathcal{P}$ covers all vertices in $G$. Further, as $\open = \current = \emptyset$, for each set $P$ in $\mathcal{P}$ there is a \textsf{Forget} node $t$ such that previously there was a record for $P$ in $\current$ but now there is no record for $P$ in both $\current$ and $\open$. Correctness follows as this can, by construction, only happen if the $P$ is fair and, at the latest from this point on, the contribution of all (non)-edges incident to $P$ is accounted for in $\cost$.    

Last, we argue that the runtime of algorithm Å is as stated.
Any record in $\mathcal{M}_t$ is of the form $\{(\open,\current,\cost)\}$. Let us analyze how many different ways there are of writing such a record. The set $\open$ is a multiset, which can be modeled as a vector counting the number of entries of each type. A cluster in the past is characterized by its size (there are at most $24\tw$ options) and how many vertices of each color it has ($\kappa$ integers, each between $0$ and $\max(24\tw,\tc)$). Thus, there are at most $(\max(24\tw,\tc))^{\kappa+1}$ distinct entries in $\open$ with at most $n$ of each type. 
Next, $\current$ is a set of at most $|B_t| \le w$ entries. There are less than $w^w$ ways in which they describe a partition of $B_t$ and, by the same reasoning as above, for each of these there are at most $((\max(24\tw,\tc))^{\kappa+1})^w$ ways to set the respective values for $\size$ and vertices of each color.
As each (non)-edge contributes $0$ or $1$ to the cost, there are $n^2+1$ possible ways to set the cost.
In total, the number of records at each table is at most
\[
    n^{(\max(24\tw,\tc))^{\kappa+1}}\cdot w^w \cdot ((\max(24\tw,\tc))^{\kappa+1})^w \cdot (n^2+1).
\] 
The table at each of the $\bigoh(n)$ nodes is computed from the table(s) of its child node(s) in polynomial time in the size of the child table(s). This time dominates the computation of the tree decomposition, so the total runtime of Å is at most 
$n^{\bigoh((\max(24\tw,\tc))^{\kappa+1})},
$
using $\max(24\tw,\tc)=\bigoh(n)$ and $w\in \bigoh(\tw)$.
\fi
\qed\end{proof}

\section{Tractability via Treewidth for Fairlets of Size at Most 2}
While the question of whether \fcc in general is fixed-parameter tractable w.r.t.\ the treewidth remains open, for the special case of fairlet size $\tc\leq 2$ we can answer it affirmatively.

In particular, the algorithm underlying Theorem~\ref{thm:xptw} stores a large amount of information about clusters in the past graph. However, for $\tc\leq 2$, we show that, without loss of generality, clusters intersecting the past and future vertices but not the current bag $B_t$ are trivial and easy to keep track of. Towards this, we establish the following lemma:

\begin{lemma}
    \label{lem:SplitoffPairsOfVertices}
Every \yesinstance of \fcc with fairlet size $\tilde{c}=2$ on a graph $G$ admits a solution $\mathcal{P}$ with the following property: each cluster in $\mathcal{P}$ of size greater than $2$ induces a connected subgraph of $G$.
\end{lemma}

\begin{proof}
Consider an arbitrary solution $\mathcal{P'}$ to the instance and suppose it contains a cluster $Z$ not satisfying the claimed property. Let $A\subseteq Z$ be maximal such that $A$ induces a connected component in $G$ and let $r$ and $b$ be the numbers of vertices of the first color (\emph{red}) and the second color (\emph{blue}) in $A$, respectively. Without loss of generality, assume $r \ge b$. Consider an arbitrary set $R\subset A$ of $r-b$ red vertices from $A$ and let $\mathcal{P}$ be the fair clustering obtained from $\mathcal{P'}$ by replacing $Z$ by $Z_1 \coloneqq A \setminus R$ and $Z_2 \coloneqq Z\setminus Z_1$. Now $Z_1$, satisfies the lemma's requirement, while $Z_2$ may still not satisfy it. We will show that the cost of $\mathcal{P}$ is at most the cost of $\mathcal{P'}$. Assuming that this holds, we can exhaustively apply the above modification to obtain a solution satisfying the lemma's requirement, where termination is guaranteed as each modification strictly increases the number of clusters.

Thus, it remains to show $\cost(P) \le \cost(P')$, where $\cost(\cdot)$ denotes the cost of the respective clustering. Just like in the proof of Lemma~\ref{lem:MaximumClustersizeTreewidth}, the former can be obtained from the latter as follows: 
\begin{align*}
\cost(\mathcal{P})=\cost(\mathcal{P}') & +  |\{\{v,w\}~|~v\in Z_1, w\in Z_2, \set{v,w}\in E(G)\} | \\
 & - |\{\{v,w\}~|~v\in Z_1, w\in Z_2, \set{v,w}\not \in E(G)\} |.
\end{align*}
By construction, the set in the first line has size at most $|Z_1|\cdot |R|$ and the set in the second line has size at least $|Z_1| \cdot |Z\setminus A|$. As $Z$ is fair, the excess of red vertices in $A$ equals the excess of blue vertices in $Z\setminus A$. Thus, there are at least $r-b = |R|$ blue vertices in $Z\setminus A$ and thereby $\cost(\mathcal{P})\le \cost(\mathcal{P}')$. \qed
\end{proof}

We call a solution \emph{nice} if it satisfies the property in Lemma~\ref{lem:SplitoffPairsOfVertices}. We can now restrict our focus on identifying whether an instance admits any nice solution.

\thmfpttw*

\begin{proof}
We first discuss the case $\tc=2$, for which we modify the algorithm used for Theorem~\ref{thm:xptw}. Instead of keeping track of all possibilities for past clusters, we only keep ones that can occur in nice solutions. Notably, we disregard records for which there is an entry with assumed size larger than $2$ in $\open$ as every vertex added later to complete the respective cluster is necessarily not adjacent to any of the current ones.
We realize this by adjusting the procedure in \textsf{forget} nodes: we no longer add any record for which $\open$ consists anything but entries of type $(2,(1,0),v)$ or $(2,(0,1),v)$, for any vertex $v$.
Crucially, after this change there is still a record $(\emptyset,\emptyset,B')$ with $B'\le B$ for any nice solution to a \yesinstance.

This dramatically reduces the size of the tables: Now, the set $\open$ consists of $0$ to $n$ entries and each has one of two types, so there are $O(n^2)$ options for the $\open$ entry in a record. The runtime of the algorithm for the case of $\tc=2$ (and hence $\kappa = 2)$ reduces to at most
\[
    O(n^2)\cdot (5\tw+4)^{5\tw+4} \cdot ((\max(24\tw,\tc))^{\kappa+1})^{5\tw+4} \cdot (n^2+1) =\tw^{\bigoh(\tw)} \bigoh(n^4).
\]
The above running time is not claimed to be tight---for instance, it is not difficult to observe that one does not need to store the number of isolated blue and red vertices as two numbers in the past---however, it suffices for our classification result. Finally, while further adapting the algorithm to the fairness-oblivious case of $\tc=1$ is rather straightforward, for our purposes it suffices to note that there is a polynomial-time reduction from that case to \textsc{FCC} with $\tc=2$ which only increases the treewidth by at most a factor of $2$~\cite[Section~5.4]{abs200203508}---in particular, the reduction changes the graph by essentially duplicating every vertex. 
 \qed
\end{proof}

A surprising consequence of this result is that it settles the (to the best of our knowledge) open question, whether fairness-oblivious \textsc{Correlation clustering} (that is, \textsc{Cluster Editing}) is fixed-parameter tractable w.r.t.\ the treewidth (as these problems are equivalent to \fcc with $\tc = 1$). 
\begin{corollary}
    Fairness-oblivious \textsc{Correlation Clustering} (that is, \textsc{Cluster Editing}) is in $\FPT$ w.r.t.\ the treewidth of $G$.
\end{corollary}

\section{Tractability via Treedepth and the Fairlet Size}

This section provides our final contribution towards the complexity map of \textsc{FCC}:

\thmtd*

Before proceeding to the proof of Theorem~\ref{thm:td}, we first show how to solve instances where every connected component has bounded size. Such instances will arise during our reduction procedure, and unlike in typical graph problems we cannot solve each connected component entirely independently of the rest.

\iflong
\begin{lemma}
\fi
\ifshort
\begin{lemma}[$\spadesuit$]
\fi
\label{lem:solve_for_small_components}
      Let $\inbound, \solbound, B \in \N$ and consider a colored graph $G$ with $\kappa$ colors such that each connected component in $G$ has size at most $\inbound$. There is an algorithm which determines whether $G$ admits a fair clustering of cost at most $B$ such that each cluster has size at most $\solbound$ and runs in time $f(\kappa+\inbound+\solbound) \cdot |V(G)|^{\bigoh(1)}$ for some computable function $f$.      
\end{lemma}

\ifshort
\begin{proof}[Proof Sketch]
We encode the instance as an integer linear program (ILP) with at most $f'(\kappa+\inbound+\solbound)$ variables and constraints for a computable function $f'$. 
      Such an ILP can be solved in the desired time, e.g., by using the result by Lenstra~\cite{Lenstra_1983} and its subsequent improvements \cite{Frank_Tardos_1987,Kannan_1987}. Since there is only a bounded number of pairwise non-isomorphic connected components in $G$ and each can be partitioned in only a bounded number of ways, we can use variables to capture how each of the isomorphism classes of components is partitioned by the solution. Moreover, there is a bounded number of ways how each fair cluster in the solution can be composed from the parts of the connected components, and we can also use variables to capture this. Finally, one adds constraints which ensure that the budget is not exceeded and that the numbers of parts form a partitioning.
\qed
\end{proof}
\fi

\iflong
\begin{proof}
      We encode the instance as an integer linear program (ILP) such that there are at most $f'(\kappa+\inbound+\solbound)$ variables and constraints for a computable function $f'$. 
      Such an ILP can be solved in the desired time, e.g., by using the result by Lenstra~\cite{Lenstra_1983} and its subsequent improvements \cite{Frank_Tardos_1987,Kannan_1987}.

      We say that the \emph{type} of a set $M$ of vertices each vertex colored in one of $\kappa$ colors, is the vector $t = (t_1,\ldots, t_\kappa)\in (\set{0}\cup \N)^\kappa$ such that $M$ consists of precisely $t_i$ vertices of each color $i$. For $\ell \in \N$, we define $T_\ell \coloneqq (\set{0}\cup [\ell])^\kappa$, which covers all types of sets of size at most $\ell$. 
    We describe three sets of variables for the ILP instance.
    The first set, $\Comps$, contains a variable $\comp_t$ for each vector $t\in T_\solbound$. Intuitively, the variable $\comp_t$ will describe how many connected components \emph{within a cluster} of a hypothetical solution will have type $t$. 
    
    The next set, $\Cuts$, contains variables describing how the connected components in $G$ are partitioned into connected components within clusters by a hypothetical solution.
    Partition the set of connected components in $G$ such that two connected components share a class if they are isomorphic (respecting the vertex colors), and let $R$ be the set of such classes. As all connected components have size at most $\inbound$, we have $|R| \le 2^{\binom{\inbound}{2}} \inbound^\kappa$.
 For each class $C$ in $R$, let $Q_C$ be the set of partitions of the vertex set in any graph represented by $C$. Then $|Q_C| \le \inbound^{\inbound}$.
    Let $\Cuts$ consist of one variable $\cut_{C,q_C}$ for each $C\in R$ and each $q_C\in Q_C$. Intuitively, it represents how many connected components in class $C$ are partitioned as described by $q_C$ by the clustering of a hypothetical solution.
  For each $\cut_{C,q_C} \in \Cuts$ let $\cost(\cut_{C,q_C})$ refer to the number of edges connecting vertices in distinct sets of the respective partition plus the total number of non-edges in the graphs induced by each set of the partition $q_C$. 
      Also, for each $\comp_t \in \Comps$, let $\counter(\cut_{C,q_C}, \comp_t)$ denote the number of sets of type $t$ in the partition described by $t_C$ on any graph in class $C$.    
      
    The third set, $\Sols$, contains variables describing from which types of connected components the clusters in a hypothetical solution are built. Let $S = (\set{0\cup [\solbound]})^{|T_\solbound}|)$. Then, let $\Sols'$ consists of a variable $\sol_s$ for each $s\in S$. Intuitively, $\sol_s$ with $s=(s_1, \ldots, s_{|T_\solbound|})$ states, how many of the clusters in a hypothetical solution have shape $s$.
    A cluster has shape $s$ if, for each $j\in [|T_\solbound|]$, it consists of precisely $s_j$ connected components of type $t_j$ (the $j^{\text{th}}$ type in $T_\solbound$). Let $\Sols \subseteq \Sols'$ be the restriction to shapes such that each shape describes a \emph{fair} cluster, that is, a set of vertices constructed as follows would be fair: it starts empty and for each $j\in [|T_\solbound|]$ and each color $i$, we add $s_j \cdot t_j[i]$ vertices of color $i$.
      For each $\comp_{t_j}\in\Comps$ and $\sol_s\in \Sols$, let $\counter(\sol_s, \comp_{t_j}) \coloneqq s_j$ denote the number of connected components of type $t_j$ in shape $s$.

    We remark that 
    $|\Comps| \le (\solbound+1)^\kappa$,  
    $|\Cuts| \le |R|\inbound^\inbound \le 2^{\binom{\inbound}{2}}\inbound^{\kappa+\inbound}$, and 
    $|\Sols| \le (\solbound+1)^{(\solbound+1)^\kappa}$.
    We next add the following constraints to the ILP.
    First, we require that each variable is at least $0$.
    For each $C \in R$, let $|C|$ denote the number of connected components in $G$ of class $C$ and add the constraint
    \begin{equation}
       |C|=\sum_{q_C\in Q_C} \cut_{C, q_C}.
    \end{equation}
    This ensures that the variables in $\Cuts$ describe how the connected components in $G$ are sub-partitioned by a hypothetical solution.
    Next, we ensure that splitting the connected components in $G$ as described by $\Cuts$ produces the number of component types as stated by $\Comps$. For each $\comp\in \Comps$ we require
     \begin{align}
            \comp &= \sum_{\cut\in \Cuts} \cut \cdot \counter(\cut, \comp).\label{eq:comp_cuts}
      \end{align} 

      To ensure that from the components with the types in $\Comps$ we can build a fair clustering as described by $\Sols$ we require for each $\comp \in \Comps$ that
      \begin{equation}\label{eq:comp_sol}
            \comp = \sum_{\sol\in \Sols} \sol \cdot \counter(\sol, \comp).
      \end{equation} 
      Furthermore, for each $\sol\in \Sols$, let $\cost(\sol)$ denote the number of pairs of vertices in distinct components in $\sol$, for example, if $\sol$ consists of three connected components with $4,5,$ and $6$ vertices, respectively, then $\cost(\sol) = 4\cdot 5 + 4 \cdot 6 + 5 \cdot 6 = 74$. With this, we can ensure that the cost of the described clustering is within the budget by requiring
      \begin{align}
            B &\ge  \sum_{\cut\in \Cuts} \cut \cdot \cost(\cut)+ \sum_{\sol\in\Sols} \sol\cdot \cost(\sol). \label{eq:total_cost}
      \end{align} 
      Note that the ILP can be constructed (and solved) in time in $f(c+\inbound+\solbound)\cdot |V(G)|^{O(1)}$.

      It remains to argue that the ILP has a solution if and only if $G$ admits a fair clustering with cost at most $B$ such that every cluster in the solution has size at most $\solbound$. 
      Suppose there is such a clustering $\mathcal{P}$. Let the value of each $\comp\in\Comps, \cut\in\Cuts,$ and $\sol\in \Sols$ be as induced by $\mathcal{P}$ given the above interpretation of each variable.
      Note that this satisfies the constraints in Eqs.~(\ref{eq:comp_cuts})~and~(\ref{eq:comp_sol}) for all choices of $\comp\in \Comps$.
      For the constraint in Eq.~(\ref{eq:total_cost}), recall that the cost of partition $\mathcal{P}$ is the number of edges cut by $\mathcal{P}$ plus the total number of non-edges within clusters of $\mathcal{P}$.
      There are two categories of the latter: non-edges within connected components of $S$ and non-edges between distinct connected components of $S$ that share a cluster in $\mathcal{P}$. This latter cost is precisely $\cost(\sol)$ for every solution cluster of type $\sol$, while the other costs are entirely captured by $\sum_{\cut\in \Cuts} \cut \cdot \cost(\cut)$. Hence, Eq.~(\ref{eq:total_cost}) is satisfied.

      For the other direction, assume that the ILP has a solution. Cutting edges in $G$ as described by the variables in $\Cuts$ and reassembling as described by $\Sols$ yields a (not necessarily unique) fair clustering of $V(G)$, with no cluster larger than~$\solbound$. Eq.~(\ref{eq:total_cost}) implies that the cost of this clustering is at most $B$.
\qed\end{proof}
\fi

\thmtd*
\begin{proof}
      Consider an instance $(G, B, \chi)$ of $\textsc{FCC}$ such that $G$ has treedepth $\td$ and the fairlet size is $\tc$.
      We construct an equivalent instance $(G', B',\chi)$ such that, for a computable function $f$, all connected components in $G'$ are smaller than $f(\tilde{c},\td)$.
      As we ensure that $G'$ is a subgraph of $G$, we have $\tw(G') \le \td(G') \le \td(G)$ and thus, by Lemma~\ref{lem:MaximumClustersizeTreewidth}, it suffices to decide whether there is a solution with clusters of size at most $\solbound \coloneqq \max(24\td,\tilde{c})$.
      Such an instance $(G',B',\chi)$ can then be solved by invoking Lemma~\ref{lem:solve_for_small_components}, yielding fixed-parameter tractability w.r.t.\ $\tc+\td$.

      To construct $(G', B',\chi)$, consider a rooted forest $F$ over vertices $V(G)$ which witnesses the treedepth $\td$ of $G$ (i.e., the height of each tree in $F$ is at most $\td$ and all vertices connected by edges in $G$ have ancestor-descendant relationship in $F$). 
      For every vertex $v\in V(G)$, let its \emph{type} be recursively defined from leaves to roots by the color of $v$, the set of its ancestors in $F$ to which it is adjacent in $G$, and the multiset of types among its children.
      Perform the following procedure for every tree $T$ in $F$: Start at the bottom layer (the leaves). From this layer, select $\solbound$ vertices of each type (or all vertices with that type if it contains fewer than $\solbound$ vertices). 
      For each vertex $v$ that is not selected, remove all edges between $v$ and its ancestors in $T$ (in both $G$ and $F$) and update all types accordingly. For each edge removed this way, reduce $B$ by one.
      Repeat this procedure on $T$ until reaching the root layer. 
      If at any point the budget $B$ drops below $0$, reject the instance. This is correct as we remove an edge if we know that we cannot keep all neighbors of some vertex $v$ in the same cluster as $v$, making at least this number of cuts necessary.
      If the instance has not been rejected after all trees in the forest are processed, let $G'$ be the updated graph $G$ and $B'$ be the updated budget. 
      Observe that $G'$ and $B'$ can be constructed in \FPT-time w.r.t.\ $\tc+\td$.      

      We argue that $(G', B', \chi)$ is a \yesinstance if and only so is $(G,B,\chi)$.
      Note that $G$ and $G'$ are identical except that $G$ has $B-B'$ additional edges. 
      Suppose a fair clustering $\mathcal{P}$ of $V(G)$ witnesses $G'$ to be a \yesinstance. Then $\mathcal{P}$ is a fair clustering in $G$, and the cost of $\mathcal{P}$ in $G$ is at most its cost in $G'$ (i.e., $B'$) plus the number of edges in $E(G)\setminus E(G')$ that contribute to the cost (i.e., $B-B'$). Hence, the cost of $\mathcal{P}$ in $G$ is at most $B$, witnessing that $(G,B,\chi)$ is a \yesinstance.
      
      For the other direction, assume $(G, B, \chi)$ is a \yesinstance, then it is witnessed by a clustering $\mathcal{P}$ with clusters of size at most $\solbound$.
      Let $\widetilde{E} = E(G) \setminus E(G')$, and call all edges in $\widetilde{E}$ \emph{problematic} that connect vertices within the same cluster in $\mathcal{P}$.
      We show that there is a fair clustering $\mathcal{P}'$ of $V(G)$ without problematic edges that has at most the cost of $\mathcal{P}$ on $G$.
      Then, as each edge in $\widetilde{E}$ connects vertices of distinct clusters in $\mathcal{P}'$, the cost of $\mathcal{P}'$ on $G'$ is at most $B - |\widetilde{E}| = B'$. 
      We obtain $\mathcal{P}'$ by adapting $\mathcal{P}$ for each problematic edge by the following recursive procedure, which is started once with each vertex $r$ that is the root of some tree $T$ in $F$.
      Assume the procedure is currently visiting a vertex~$v$.
      If $v$ is not incident to any problematic edges, recurse on each child of $v$ in $T$. 
      Otherwise, before recursing to its children, process the problematic edges from top to bottom, that is, if $w$ has a higher layer than $w'$ then a problematic edge $\set{v,w}$ would be processed before $\set{v,w'}$. Ties are resolved arbitrarily. 
      To process a problematic edge $\set{v,w}$, note that the non-existence of $\set{v,w}$ in $G'$ implies that $v$ is adjacent to at least $\solbound$ vertices $u$ with the same type as $w$ in $G$ such that $\set{v,u}$ is not problematic. Pick any of these vertices $u$ such that $u$ does not yet share a cluster with $v$ in $\mathcal{P}$ and let $V_u$ and $V_w$ denote the sets of vertices in the subtrees of $T$ rooted at $u$ and $w$, respectively.
      Note that having the same type implies that $G[u]$ and $G[w]$ are isomorphic even with respect to vertex colors and their neighborhood to the ancestors of $v$. Hence, if in $\mathcal{P}$ we swap each vertex in $V_u$ with its counterpart in $V_w$, we obtain a fair partition with one fewer problematic edge, since $v$ and $w$ do not share a cluster anymore. 
      Even though some problematic edges might have been replaced by others, observe that none of these new problematic edges is incident to any vertex which we already processed.
      By performing the full recursion, we receive a fair partition $\mathcal{P'}$ with no problematic edges and cost at most $B - |\widetilde{E}| = B'$ in $G'$ as desired.

      It remains to bound the size of each connected component in $G'$.
      For any tree $T$ and layer $i$, let $t_i$ denote the number of types vertices in that layer after $T$ is processed.
      There are $\kappa2^{\td}$ types of leaves, so $t_0 \le \kappa 2^{\td}$ and, after the processing, there are at most $t_0\solbound$ vertices in this layer. 
      For each layer $i \in [\td]$, we have $t_i \le \kappa 2^{\td-i}\cdot (\solbound+1)^{t_{i-1}}$ types and at most $t_i \solbound$ vertices.
      By recalling that $\solbound \le \max(24\td, \tc)$ the size of~$T$ after the processing is upper bounded by a computable function in $\tc+\td$.
      The same argument applies to connected components which are created by being cut off from the main tree.
\qed\end{proof}


\section{Concluding Remarks}
Our results identify two routes to tractability for \textsc{Fair Correlation Clustering}: Either we may use Theorem~\ref{thm:vcn} to solve instances with large fairlets but a simple graph structure, or Theorems~\ref{thm:fpttw}-\ref{thm:td} for instances with a richer graph structure but small fairlets. While the fixed-parameter tractability of the problem w.r.t.\ treewidth on instances with fairlets of size $2$ (Theorem~\ref{thm:fpttw}) also settles an analogous question for \textsc{Correlation Clustering}, it simultaneously leads to the main open question arising from our work: is \textsc{Fair Correlation Clustering} fixed-parameter tractable w.r.t.\ treewidth plus the fairlet size?

\medskip
\noindent \textbf{Acknowledgments.} The authors acknowledge support from the Austrian Science Fund (Project 10.55776/Y1329).

\bibliographystyle{splncs04}
\bibliography{fair_clustering}

\end{document}

\end{document}